\documentclass[letterpaper,11pt]{article}
\usepackage[utf8]{inputenc}
\usepackage{times}
\usepackage{comment}
\usepackage{preamble}
\allowdisplaybreaks

\newcommand{\nope}[1]{}

\newcommand{\abs}[1]{\left| #1 \right|}

\renewcommand{\epsilon}{\varepsilon}

\newcommand{\RR}{\mathbb{R}}

\newcommand{\NULL}{\mathsf{NULL}}

\newcommand{\id}{\mathbb{I}}

\newcommand{\TLap}{\mathrm{TLap}}

\newcommand{\Score}{\textsc{Score}}

\newcommand{\pcount}{\mathrm{Count}}

\newcommand{\DPASE}{\mathrm{DPASE}}
\newcommand{\DPASEB}{\mathrm{DPASEB}}
\newcommand{\DPESE}{\mathrm{DPESE}}
\newcommand{\GAPMAX}{\mathrm{GAP\text{-}MAX}}

\newcommand{\tnote}[1]{}%\textcolor{red}{[Thomas: #1]}}
\newcommand{\vnote}[1]{}%\textcolor{orange}{[Vikrant: #1]}}

\usepackage[style=alphabetic,backend=bibtex,maxalphanames=10,maxbibnames=20,maxcitenames=10,giveninits=true,doi=false,url=true]{biblatex}
\newcommand*{\citet}[1]{\AtNextCite{\AtEachCitekey{\defcounter{maxnames}{2}}}\textcite{#1}}
\newcommand*{\citetall}[1]{\AtNextCite{\AtEachCitekey{\defcounter{maxnames}{999}}}\textcite{#1}}
\newcommand*{\citep}[1]{\citep{#1}}

\usepackage{hyperref}

\title{Privately Learning Subspaces}
\iftrue
\author{Vikrant Singhal
    \thanks{Northeastern University. Part of this work was done during an internship at IBM Research -- Almaden.~\dotfill~\texttt{singhal.vi@northeastern.edu}}
    \and
    Thomas Steinke\thanks{Google Research, Brain Team. Part of this work was done at IBM Research --
    Almaden.~\dotfill~\texttt{subspace@thomas-steinke.net}}}
\date{}
\fi
\begin{document}
\maketitle

\begin{abstract}
    Private data analysis suffers a costly curse of dimensionality. However, the data often has an underlying low-dimensional structure. For example, when optimizing via gradient descent, the gradients often lie in or near a low-dimensional subspace. If that low-dimensional structure can be identified, then we can avoid paying (in terms of privacy or accuracy) for the high ambient dimension. 
    
    We present differentially private algorithms that take input data sampled from a low-dimensional linear subspace (possibly with a small amount of error) and output that subspace (or an approximation to it). These algorithms can serve as a pre-processing step for other procedures.
\end{abstract}
%\newpage

\section{Introduction}

Differentially private algorithms generally have a poor dependence on the dimensionality of their input. That is, their error or sample complexity grows polynomially with the dimension. For example, for the simple task of estimating the mean of a distribution supported on $[0,1]^d$, we have per-coordinate error $\Theta(\sqrt{d}/n)$ to attain differential privacy, where $n$ is the number of samples. In contrast, the non-private error is $\Theta(\sqrt{\log(d)/n})$.

This cost of dimensionality is inherent \cite{BunUV14,SteinkeU17a,DworkSSUV15}. \emph{Any} method with lower error is susceptible to tracing attacks (a.k.a.~membership inference attacks). However, these lower bounds only apply when the data distribution is ``high-entropy.'' This leaves open the posssibility that we can circumvent the curse of dimensionality when the data has an underlying low-dimensional structure.

Data often does possess an underlying low-dimensional structure. For example, the gradients that arise in deep learning tend to be close to a low-dimensional subspace \cite{AbadiCGMMTZ16,LiZTSG17,GurAriRD18,LiFLY18,LiGZCB19,ZhouWB20,FengT20}. Low dimensionality can arise from meaningful relationships that are at least locally linear, such as income versus tax paid. It can also arise because we are looking at a function of data with relatively few attributes.

A long line of work \cite[etc.]{BlumLR08,HardtT10,HardtR10,Ullman15,BlasiokBNS19,BassilyCMNUW20,ZhouWB20,KairouzRRT20} has shown how to exploit structure in the data to attain better privacy and accuracy. However, these approaches assume that this structure is known \emph{a priori} or that it can be learned from non-private sources. This raises the question:
\begin{quote}
    Can we learn low-dimensional structure from the data subject to differential privacy?
\end{quote}
We consider the simple setting where the data lies in $\mathbb{R}^d$ but is in, or very close to a linear subspace, of dimension $k$. We focus on the setting where $k \ll d$ and we develop algorithms whose sample complexity does not depend on the ambient dimension $d$; a polynomial dependence on the true dimension $k$ is unavoidable.

Our algorithms identify the subspace in question or, if the data is perturbed slightly, an approximation to it. Identifying the subspace structure is interesting in its own right, but it also can be used as a pre-processing step for further analysis -- by projecting to the low-dimensional subspace, we ensure subsequent data analysis steps do not need to deal with high-dimensional data.

\subsection{Our Contributions: Privately Learning Subspaces -- Exact Case}

We first consider the exact case, where the data $X_1, \cdots, X_n \in \mathbb{R}^d$ are assumed to lie in a $k$-dimensional subspace (rather than merely being near to it) -- i.e., $\mathsf{rank}\left(A\right) = k$, where $A = \sum_i^n X_i X_i^T \in \mathbb{R}^{d \times d}$. In this case, we can also recover the subspace exactly.

However, we must also make some non-degeneracy assumptions. We want to avoid a pathological input dataset such as the following. Suppose $X_1, \cdots, X_k$ are linearly independent, but $X_k=X_{k+1}=X_{k+2}=\cdots=X_n$. While we can easily reveal the repeated data point, we cannot reveal anything about the other points due to the privacy constraint. 

A natural non-degeneracy assumption would be to assume that the data points are in ``general position'' -- that is, that there are no non-trivial linear dependencies among the data points. This means that \emph{every} set of $k$ data points spans the subspace or, equivalently, no subspace of dimension $k-1$ contains more than $k-1$ data points. This is a very natural assumption -- if the data consists of $n$ samples from a continuous distribution on the subspace, then this holds with probability $1$. We relax this assumption slightly and assume that no subspace of dimension $k-1$ contains more than $\ell$ data points. We also assume that all points are non-zero. Note that we define subspaces to pass through the origin; our results can easily be extended to affine subspaces.

\begin{theorem}[Main Result -- Exact Case]\label{thm:intro-main-exact}
For all $n,d,k,\ell \in \mathbb{N}$ and $\varepsilon,\delta>0$ satisfying $n \ge O\left(\ell + \frac{\log(1/\delta)}{\varepsilon}\right)$, there exists a randomized algorithm $M : \mathbb{R}^{d \times n} \to \mathcal{S}_d^k$ satisfying the following. Here $\mathcal{S}_d^k$ denotes the set of all $k$-dimensional subspaces of $\mathbb{R}^d$.
\begin{itemize}
    \item $M$ is $(\varepsilon,\delta)$-differentially private with respect to changing one column of its input.
    \item Let $X = (X_1, \cdots, X_n) \in \mathbb{R}^{d \times n}$.
    Suppose there exists a $k$-dimensional subspace $S_* \in \mathcal{S}_d^k$ that contains all but $\ell$ of the points -- i.e., $|\{i \in [n] : X_i \in S_*\}| \ge n -\ell$.
    Further suppose that any $(k-1)$-dimensional subspace contains at most $\ell$ points -- i.e., for all $S \in \mathcal{S}_d^{k-1}$, we have $|\{i \in [n] : X_i \in S\}| \le \ell$. 
    Then $\pr{}{M(X)=S_*}=1$.
\end{itemize}
\end{theorem}

The parameter $\ell$ in Theorem \ref{thm:intro-main-exact} can be thought of as a robustness parameter. Ideally the data points are in general position, in which case $\ell=k-1$. If a few points are corrupted, then we increase $\ell$ accordingly; our algorithm can tolerate the corruption of a small constant fraction of the data points.
Theorem \ref{thm:intro-main-exact} is optimal in the sense that $n \ge \Omega\left(\ell + \frac{\log(1/\delta)}{\varepsilon}\right)$ samples are required.

\subsection{Our Contributions: Privately Learning Subspaces -- Approximate Case}

Next we turn to the substantially more challenging approximate case, where the data $X_1, \cdots, X_n \in \mathbb{R}^d$ are assumed to be close to a $k$-dimensional subspace, but are not assumed to be contained within that subspace. Our algorithm for the exact case is robust to changing a few points, but very brittle if we change all the points by a little bit. Tiny perturbations of the data points (due to numerical errors or measurement imprecision) could push the point outside the subspace, which would cause the algorithm to fail. Thus it is important to for us to cover the approximate case and our algorithm for the approximate is entirely different from our algorithm for the exact case. 

The approximate case requires us to precisely quantify how close the input data and our output are to the subspace and we also need to make quantitative non-degeneracy assumptions. It is easiest to formulate this via a distributional assumption. We will assume that the data comes from a Gaussian distribution where the covariance matrix has a certain eigenvalue gap. This is a strong assumption and we emphasize that this is only for ease of presentation; our algorithm works under weaker assumptions. Furthermore, we stress that the differential privacy guarantee is worst-case and does not depend on any distributional assumptions.

We assume that the data is drawn from a multivariate Gaussian $\mathcal{N}(0,\Sigma)$. Let $\lambda_1(\Sigma) \ge \lambda_2(\Sigma) \ge \cdots \ge \lambda_d(\Sigma)$ be the eigenvalues of $\Sigma \in \mathbb{R}^{d \times d}$. We assume that there are $k$ large eigenvalues $\lambda_1(\Sigma), \cdots, \lambda_k(\Sigma)$ -- these represent the ``signal'' we want -- and $d-k$ small eigenvalues $\lambda_{k+1}(\Sigma), \cdots, \lambda_d(\Sigma)$ --  these are the ``noise''. Our goal is to recover the subspace spanned by the eigenvectors corresponding to the $k$ largest eigenvalues $\lambda_1(\Sigma), \cdots, \lambda_k(\Sigma)$.
Our assumption is that there is a large \emph{multiplicative} gap between the large and small eigenvalues. Namely, we assume $\frac{\lambda_{k+1}(\Sigma)}{\lambda_{k}(\Sigma)} \le \frac{1}{\mathsf{poly}(d)}$. 

\begin{theorem}[Main Result -- Approximate Case]\label{thm:intro-main-approx}
    For all $n,d,k \in \mathbb{N}$ and $\alpha,\gamma,\varepsilon, \delta > 0$ satisfying
    \[n \!\ge\! \Theta\!\left(\!\frac{k\log(1/\delta)}{\varepsilon} \!+\!
        \frac{\ln(1/\delta)\ln(\ln(1/\delta)/\eps)}{\eps}\!\right)
    ~\text{and}~
    \gamma^2 \!\le\! \Theta\!\left(\!\frac{\eps\alpha^2n}{d^2k\log(1/\delta)}\!\cdot\!\min\!\left\{\!\frac{1}{k},\!
        \frac{1}{\log(k\log(1/\delta)/\varepsilon)}\!\right\}\!\right)\!,\]
    there exists an algorithm $M : \mathbb{R}^{d \times n} \to \mathcal{S}_d^k$ satisfying the following. Here $\mathcal{S}_d^k$ is the set of all $k$-dimensional subspaces of $\mathbb{R}^d$ represented as projection matricies -- i.e., $\mathcal{S}_d^k = \{\Pi \in \mathbb{R}^{d \times d} : \Pi^2=\Pi=\Pi^T, \mathsf{rank}(\Pi)=k\}$.
\begin{itemize}
    \item $M$ is $(\varepsilon,\delta)$-differentially private with respect to changing one column of its input.
    \item Let $X_1, \cdots, X_n$ be independent samples from $\mathcal{N}(0,\Sigma)$. Let $\lambda_1(\Sigma) \ge \lambda_2(\Sigma) \ge \cdots \ge \lambda_d(\Sigma)$ be the eigenvalues of $\Sigma \in \mathbb{R}^{d \times d}$. Suppose $\lambda_{k+1}(\Sigma) \le \gamma^2 \cdot \lambda_k(\Sigma)$. Let $\Pi \in \mathcal{S}_d^k$ be the projection matrix onto the subspace spanned by the eigenvectors corresponding to the $k$ largest eigenvalues of $\Sigma$. Then $\pr{}{\|M(X)-\Pi\| \le \alpha} \ge 0.7$.
\end{itemize}
\end{theorem}

The sample complexity of our algorithm $n=O(k \log(1/\delta)/\varepsilon)$ is independent of the ambient dimension $d$; this is ideal. However, there is a polynomial dependence on $d$ in $\gamma$, which controls the multiplicative eigenvalue gap. % We conjecture that this can be improved, but it cannot eliminated entirely.
This multiplicative eigenvalue gap is a strong assumption, but it is also a necessary assumption if we want the sample complexity $n$ to be independent of the dimension $d$. In fact, it is necessary \emph{even without the differential privacy constraint} \cite{CaiZ16}. That is, if we did not assume an eigenvalue gap that depends polynomially on the ambient dimension $d$, then it would be impossible to estimate the subspace with sample complexity $n$ that is independent of the ambient dimension $d$ even in the non-private setting.

Our algorithm is based on the subsample and aggregate framework \cite{NissimRS07} and a differentially private histogram algorithm. These methods are generally quite robust and thus our algorithm is, too. For example, our algorithm can tolerate $o(n/k)$ input points being corrupted arbitrarily. \vnote{We don't exactly prove that in the technical section.} We also believe that our algorithm's utility guarantee is robust to relaxing the Gaussianity assumption. All that we require in the analysis is that the empirical covariance matrix of a few samples from the distribution is sufficiently close to its expectation $\Sigma$ with high probability.

\subsection{Related Work}

To the best of our knowledge, the problem of privately learning subspaces, as we formulate it, has not been studied before. However, a closely-related line of work is on Private Principal Component Analysis (PCA) and low-rank approximations. We briefly discuss this extensive line of work below, but first we note that, in our setting, all of these techniques have a sample complexity $n$ that grows polynomially with the ambient dimension $d$. Thus, they do not evade privacy's curse of dimensionality. However, we make a stronger assumption than these prior works -- namely, we assume a large multiplicative eigenvalue gap. (Many of the prior works consider an \emph{additive} eigenvalue gap, which is a weaker assumption.)

There has been a lot of interest in Private PCA, matrix completion, and low-rank approximation. One motivation for this is the infamous Netflix prize, which can be interpreted as a matrix completion problem. The competition was cancelled after researchers showed that the public training data revealed the private movie viewing histories of many of Netflix's customers \cite{NarayananS06}. Thus privacy is a real concern for matrix analysis tasks.

Many variants of these problems have been considered: Some provide approximations to the data matrix $X = (X_1, \cdots, X_n) \in \mathbb{R}^{d \times n}$; others approximate the covariance matrix $A = \sum_i^n X_i X_i^T \in \mathbb{R}^{d \times d}$ (as we do). There are also different forms of approximation -- we can either produce a subspace or an approximation to the entire matrix, and the approximation can be measured by different norms (we consider the operator norm between projection matrices). Importantly, we define differential privacy to allow one data point $X_i$ to be changed arbitrarily, whereas most of the prior work assumes a bound on the norm of the change or even assumes that only one coordinate of one vector can be changed. In the discussion below we focus on the techniques that have been considered for these problems, rather than the specific results and settings.

\citetall{DworkTTZ14} consider the simple algorithm which adds independent Gaussian noise to each of entries of the covariance matrix $A$, and then perform analysis on the noisy matrix. (%This algorithm is sometimes confusingly referred to as ``randomized response.''
In fact, this algorithm predates the development of differential privacy \cite{BlumDMN05} and was also analyzed under differential privacy by McSherry and Mironov \cite{McSherryM09} and Chaudhuri, Sarwate, and Sinha \cite{ChaudhuriSS12}.) %This requires a bound on $\|X_i\|_2$ to control sensitivity. 
This simple algorithm is versatile and several bounds are provided for the accuracy of the noisy PCA. The downside of this is that a polynomial dependence on the ambient dimension $d$ is inherent -- indeed, they prove a sample complexity lower bound of $n = \tilde\Omega(\sqrt{d})$ for any algorithm that identifies a useful approximation to the top eigenvector of $A$. This lower bound does not contradict our results because the relevant inputs do not satisfy our near low-rank assumption.

\citetall{HardtR12} and \citetall{AroraBU18} apply techniques from dimensionality reduction to privately compute a low-rank approximation to the input matrix $X$. \citetall{HardtR13} and \citetall{HardtP13} use the power iteration method with noise injected at each step to compute low-rank approximations to the input matrix $X$. In all of these, the underlying privacy mechanism is still noise addition and the results still require the sample complexity to grow polynomially with the ambient dimension to obtain interesting guarantees. (However, the results can be dimension-independent if we define differential privacy so that only one entry -- as opposed to one column -- of the matrix $X$ can be changed by $1$. This is a significantly weaker privacy guarantee.)

\citetall{BlockiBDS12} and \citetall{Sheffet19} also use tools from dimensionality reduction; they approximate the covariance matrix $A$. However, they show that the dimensionality reduction step itself provides a privacy guarantee (whereas the aforementioned results did not exploit this and relied on noise added at a later stage). \citetall{Sheffet19} analyzes two additional techniques -- the addition of Wishart noise (i.e., $YY^T$ where the columns of $Y$ are independent multivariate Gaussians) and sampling from an inverse Wishart distribution (which has a Bayesian interpretation).

\citetall{ChaudhuriSS12}, \citetall{KapralovT13}, \citetall{WeiSCHT16}, and \citetall{AminDKMV18} apply variants of the exponential mechanism \cite{McSherryT07} to privately select a low-rank approximation to the covariance matrix $A$. This method is nontrivial to implement and analyse, but it ultimately requires the sample complexity to grow polynomially in the ambient dimension. %The exponential mechanism satisfies \emph{pure} differential privacy.%, and packing lower bounds \cite{HardtT10} can be used to show that a polynomial dependence on the dimension is necessary for any algorithm satisfying pure differential privacy.

\citetall{GonemG18} exploit smooth sensitivity \cite{NissimRS07} to release a low-rank approximation to the matrix $A$. This allows adding less noise than worst case sensitivity, under an eigenvalue gap assumption. However, the sample complexity $n$ is polynomial in the dimension $d$.

\paragraph{Limitations of Prior Work}
Given the great variety of techniques and analyses that have been applied to differentially private matrix analysis problems, what is missing?
We see that almost all of these techniques are ultimately based on some form of noise addition or the exponential mechanism. With the singular exception of the techniques of Sheffet \cite{Sheffet19}, all of these prior techniques satisfy pure\footnote{Pure differential privacy (a.k.a.~pointwise differential privacy) is $(\varepsilon,\delta)$-differential privacy with $\delta=0$.} or concentrated differential privacy \cite{BunS16}. This is enough to conclude that these techniques cannot yield the dimension-independent guarantees that we seek. No amount of postprocessing or careful analysis can avoid this limitation. This is because pure and concentrated differential privacy have strong group privacy properties, which means ``packing'' lower bounds \cite{HardtT10} apply.

We briefly sketch why concentrated differential privacy is incompatible with dimension-independent guarantees. Let the input be $X_1 = X_2 = \cdots = X_n = \xi/\sqrt{d}$ for a uniformly random $\xi \in \{-1,+1\}^d$. That is, the input is one random point repeated $n$ times. If $M$ satisfies $O(1)$-concentrated differential privacy, then it satisfies the mutual information bound $I(M(X);X) \le O(n^2)$ \cite{BunS16}. But, if $M$ provides a meaningful approximation to $X$ or $A = XX^T$, then we must be able to recover an approximation to $\xi$ from its output, whence $I(M(X);X) \ge \Omega(d)$, as the entropy of $X$ is $d$ bits. This gives a lower bound of $n \ge \Omega(\sqrt{d})$, even though $X$ and $A$ have rank $k=1$.

The above example shows that, even under the strongest assumptions (i.e., the data lies exactly in a rank-$1$ subspace), any good approximation to the subspace, to the data matrix $X$, or to the covariance matrix $A = XX^T$ must require the sample complexity $n$ to grow polynomially in the ambient dimension $d$ if we restrict to techniques that satisfy concentrated differential privacy. Almost all of the prior work in this general area is subject to this restriction.

To avoid a sample complexity $n$ that grows polynomially with the ambient dimension $d$, we need fundamentally new techniques.

\subsection{Our Techniques}

For the exact case, we construct a score function for subspaces that has low sensitivity, assigns high score to the correct subspace, and assigns a low score to all other subspaces. Then we can simply apply a $\GAPMAX$ algorithm to privately select the correct subspace \cite{BunDRS18}. 

The $\GAPMAX$ algorithm satisfies $(\varepsilon,\delta)$-differential privacy and outputs the correct subspace as long as the gap between its score and that of any other subspace is larger than $O(\log(1/\delta)/\varepsilon)$. This works even though there are infinitely many subspaces to consider, which would not be possible under concentrated differential privacy. 

The simplest score function would simply be the number of input points that the subspace contains. This assigns high score to the correct subspace, but it also assigns high score to any larger subspace that contains the correct subspace. To remedy this, we subtract from the score the number of points contained in a strictly smaller subspace. That is, the score of subspace $S$ is the number of points in $S$ minus the maximum over all subspaces $S' \subsetneq S$ of the number of points contained in $S'$.

This $\GAPMAX$ approach easily solves the exact case, but it does not readily extend to the approximate case. If we count points near to the subspace, rather than in it, then (infinitely) many subspaces will have high score, which violates the assumptions needed for $\GAPMAX$ to work. Thus we use a completely different approach for the approximate case.

We apply the ``subsample and aggregate'' paradigm of \cite{NissimRS07}. That is, we split the dataset $X_1, \cdots, X_n$ into $n/O(k)$ sub-datasets each of size $O(k)$. We use each sub-dataset to compute an approximation to the subspace by doing a (non-private) PCA on the sub-dataset. Let $\Pi$ be the projection matrix onto the correct subspace and $\Pi_1, \cdots, \Pi_{n/O(k)}$ the projection matrices onto the approximations derived from the sub-datasets. With high probability $\|\Pi_j-\Pi\|$ is small for most $j$. (Exactly how small depends on the eigengap.) Now we must privately aggregate the projection matrices $\Pi_1, \cdots, \Pi_{n/O(k)}$ into a single projection matrix.

Rather than directly trying to aggregate the projection matrices, we pick a set of reference points, project them onto the subspaces, and then aggregate the projected points. We draw $p_1, \cdots, p_{O(k)}$ independently from a standard spherical Gaussian. Then $\|\Pi_j p_i - \Pi p_i \| \le \|\Pi_j - \Pi\| \cdot O(\sqrt{k})$ is also small for all $i$ and most $j$. We wish to privately approximate $\Pi p_i$ and to do this we have $n/O(k)$ points $\Pi_j p_i$ most of which are close to $\Pi p_i$. This is now a location or mean estimation problem, which we can solve privately. Thus we obtain points $\hat p_i$ such that $\|\hat p_i - \Pi p_i\|$ is small for all $i$. From a PCA of these points we can obtain a projection $\hat\Pi$ with $\|\hat\Pi-\Pi\|$ being small, as required.

Finally, we discuss how to privately obtain $(\hat p_1, \hat p_2, \cdots, \hat p_{O(k)})$ from $(\Pi_1 p_1, \cdots, \Pi_1 p_{O(k)}), \cdots,$\\$(\Pi_{n/O(k)} p_1, \cdots, \Pi_{n/O(k)} p_{O(k)})$. It is better here to treat $(\hat p_1, \hat p_2, \cdots, \hat p_{O(k)})$ as a single vector in $\mathbb{R}^{O(kd)}$, rather than as $O(k)$ vectors in $\mathbb{R}^d$. We split $\mathbb{R}^{O(kd)}$ into cells and then run a differentially private histogram algorithm.
If we construct the cells carefully, for most $j$ we have that $(\Pi_j p_1, \cdots, \Pi_j p_{O(k)})$ is in the same histogram cell as the desired point $(\Pi p_1, \cdots, \Pi p_{O(k)})$. The histogram algorithm will thus identify this cell, and we take an arbitrary point from this cell as our estimate $(\hat p_1, \hat p_2, \cdots, \hat p_{O(k)})$. The differentially private histogram algorithm is run over exponentially \vnote{infinitely?} many cells, which is possible under $(\varepsilon,\delta)$-differential privacy if $n/O(k) \ge O(\log(1/\delta)/\varepsilon)$. (Note that under concentrated differential privacy the histogram algorithm's sample complexity $n$ would need to depend on the number of cells and, hence, the ambient dimension $d$.)

The main technical ingredients in the analysis of our algorithm for the approximate case are matrix perturbation and concentration analysis and the location estimation procedure using differentially private histograms.
Our matrix perturbation analysis uses a variant of the Davis-Kahan theorem to show that if the empirical covariance matrix is close to the true covariance matrix, then the subspaces corresponding to the top $k$ eigenvalues of each are also close; this is applied to both the subsamples and the projection of the reference points.
The matrix concentration results that we use show that the empirical covariance matrices in all the subsamples are close to the true covariance matrix.
This is the only place where the multivariate Gaussian assumption arises. Any distribution that concentrates well will work.

\section{Notations, Definitions, and Background Results}\label{sec:preliminaries}

\subsection{Linear Algebra and Probability Preliminaries}

Here, we mention a few key technical results that we will
be using to prove the main theorem for the approximate
case. Throughout this document, we assume that the dimension
$d$ is larger than some absolute constant, and adopt the following
notation: for a matrix $A$ of rank $r$, we use $s_1(A)\geq\dots\geq s_r(A)$
to denote the singular values of $A$ in decreasing order, and
use $\lambda_1(A)\geq\dots\geq\lambda_r(A)$ to denote the
eigenvalues of $A$ in decreasing order; let $s_{\min}(A)$ denote
the least, non-zero singular value of $A$. We omit the parentheses
when the context is clear. We begin by stating
two results about matrix perturbation
theory. The first result says that if two matrices are close
to one another in operator norm, then their corresponding
singular values are also close to one another.

Define \[\|M\| := \sup \{\|Mx\|_2 : x \in \mathbb{R}^d, ~ \|x\|_2 \le 1\}\]
to be the operator norm with respect to the Euclidean vector
norm.

\begin{lemma}[Singular Value Inequality]\label{lem:weyl-singular}
    Let $A,B \in \RR^{d \times n}$ and let $r = \min\{d,n\}$.
    Then for $1 \leq i,j \leq r$,
    $$s_{i+j-1}(A+B) \leq s_i(A) + s_j(B).$$
\end{lemma}

The following result gives a lower bound on the least
singular value of sum of two matrices.

\begin{lemma}[Least Singular Value of Matrix Sum]\label{lem:least-singular}
    Let $A,B \in \RR^{d\times n}$. Then
    $$s_{\min}(A+B) \geq s_{\min}(A) - \|B\|.$$
\end{lemma}

%\begin{lemma}[Weyl's Inequality]\label{lem:weyl}
%    Let $A$ and $B$ be $n \times n$ hermitian matrices.
%    Suppose $\lambda_1\geq\dots\geq\lambda_n$ are eigenvalues
%    of $A$, and $\nu_1\geq\dots\geq\nu_n$ are eigenvalues
%    of $B$. Then for each $1 \leq i \leq n$,
%    $$\abs{\lambda_i - \nu_i} \leq \|A-B\|.$$
%\end{lemma}

The next result bounds the angle between the subspaces
spanned by two matrices that are close to one another.
Let $X \in \RR^{d\times n}$ have the following SVD.
$$X=
    \begin{bmatrix}
        U & U_{\bot}
    \end{bmatrix} \cdot
    \begin{bmatrix}
        \Sigma_1 & 0\\
        0 & \Sigma_2
    \end{bmatrix} \cdot
    \begin{bmatrix}
        V^T\\
        V_{\bot}^T
    \end{bmatrix}$$
In the above, $U,U_{\bot}$ are orthonormal matrices
such that $U\in \RR^{d\times r}$ and $U_{\bot}\in \RR^{d \times (d-r)}$,
$\Sigma_1,\Sigma_2$ are diagonal matrices, such that
$\Sigma_1 \in \RR^{r \times r}$ and $\Sigma_2 \in \RR^{(d-r)\times(n-r)}$,
and $V,V_{\bot}$ are orthonormal matrices, such that
$V \in \RR^{n \times r}$ and $V_{\bot} \in \RR^{n \times (n-r)}$.
Let $Z \in \RR^{d\times n}$ be a perturbation matrix,
and $\hat{X} =  X + Z$, such that $\hat{X}$ has the
following SVD.
$$\hat{X}=
    \begin{bmatrix}
        \hat{U} & \hat{U}_{\bot}
    \end{bmatrix} \cdot
    \begin{bmatrix}
        \hat{\Sigma_1} & 0\\
        0 & \hat{\Sigma}_2
    \end{bmatrix} \cdot
    \begin{bmatrix}
        \hat{V}^T\\
        \hat{V}_{\bot}^T
    \end{bmatrix}$$
In the above, $\hat{U},\hat{U}_{\bot},\hat{\Sigma}_1,\hat{\Sigma}_2,\hat{V},\hat{V}_{\bot}$
have the same structures as $U,U_{\bot},\Sigma_1,\Sigma_2,V,V_{\bot}$
respectively. Let $Z_{21} = U_{\bot}U_{\bot}^T Z VV^T$ and
$Z_{12} = UU^T Z V_{\bot}V_{\bot}^T$. Suppose $\sigma_1\geq\dots\geq\sigma_r\geq0$
are the singular values of $U^T\hat{U}$. Let
$\Theta(U,\hat{U}) \in \RR^{r\times r}$ be a diagonal matrix,
such that $\Theta_{ii}(U,\hat{U}) = \cos^{-1}(\sigma_i)$.

\begin{lemma}[$\text{Sin}(\Theta)$ Theorem \cite{CaiZ16}]\label{lem:sin-theta}
    Let $X,\hat{X},Z,Z_{12},Z_{21}$ be defined as above.
    Denote $\alpha = s_{\min}(U^T \hat{X} V)$ and
    $\beta = \|U_{\bot}^T \hat{X} V_{\bot}\|$.
    If $\alpha^2 > \beta^2 + \min\{\|Z_{12}\|^2,\|Z_{21}\|^2\}$,
    then we have the following.
    $$\|\text{Sin}(\Theta)(U,\hat{U})\| \leq
        \frac{\alpha\|Z_{21}\|+\beta\|Z_{12}\|}
        {\alpha^2-\beta^2-\min\{\|Z_{12}\|^2,\|Z_{21}\|^2\}}$$
\end{lemma}

%\begin{lemma}[Davis-Kahan
%    \footnote{\href{https://www.cs.columbia.edu/~djhsu/coms4772-f16/lectures/davis-kahan.pdf}
%    {\texttt{https://www.cs.columbia.edu/\textasciitilde{}djhsu/coms4772-f16/lectures/davis-kahan.pdf}}}]\label{lem:davis-kahan}
%    Let $A,\tilde{A},U,\tilde{U},\Lambda,\tilde\Lambda
%    \in \mathbb{R}^{d \times d}$ satisfy the following.
%    (i) $A^T=A$, $\tilde{A}^T=A$, (ii) $A=U\Lambda U^T$,
%    $\tilde{A} = \tilde{U} \tilde{\Lambda} \tilde{U}^T$,
%    (iii) $U^TU=I_d=\tilde{U}^T\tilde{U}$, (iv) $\Lambda$
%    and $\tilde\Lambda$ are diagonal matrices with entries
%    in descending order. Let $\lambda_i, \tilde{\lambda}_i$
%    denote the $i^\text{th}$ diagonal entry of $\Lambda, \tilde\Lambda$.
%    Denote $U_{a:b}, \tilde{U}_{a:b} \in \mathbb{R}^{d \times (b-a+1)}$
%    to be the matrix formed by columns $a, a+1, \cdots, b$ of $U$
%    and $\tilde{U}$ respectively (i.e., corresponding to
%    $\lambda_a \ge \lambda_{a+1} \ge \cdots \ge \lambda_b$
%    and $\tilde\lambda_a \ge \tilde\lambda_{a+1} \ge \cdots \ge \tilde\lambda_b$).
%     Then, for all $k \in [d-1]$, if $\lambda_k > \tilde\lambda_{k+1}$,
%    we have
%    \[\|\tilde{U}_{1:k}^T {U}_{k+1:d}\| \le \frac{\|\tilde{U}_{1:k}^T
%        (\tilde{A}-A) {U}_{k+1:d}\|}{\lambda_k - \tilde\lambda_{k+1}}.\]
%\end{lemma}

The next result bounds $\|\text{Sin}(\Theta)(U,\hat{U})\|$ in terms
of the distance between $UU^T$ and $\hat{U}\hat{U}^T$.

\begin{lemma}[Property of $\|\text{Sin}(\Theta)\|$ \cite{CaiZ16}]\label{lem:sin-theta-property}
    Let $U,\hat{U} \in \RR^{d \times r}$ be orthonormal
    matrices, and let $\Theta(U,\hat{U})$ be defined as above in terms
    of $\hat{U},U$. Then we have the following.
    $$\|\text{Sin}(\Theta)(U,\hat{U})\| \leq \|\hat{U}\hat{U}^T-UU^T\|
        \leq 2\|\text{Sin}(\Theta)(U,\hat{U})\|$$
\end{lemma}

The next result bounds the singular values of a matrix,
whose columns are independent vectors from a mean zero,
isotropic distribution in $\R^d$. We first define the
sub-Gaussian norm of a random variable.

\begin{definition}
    Let $X$ be a sub-Gaussian random variable. The sub-Gaussian
    norm of $X$, denoted by $\|X\|_{\psi^2}$, is defined as,
    $$\|X\|_{\psi^2} = \inf\{t>0: \ex{}{\exp(X^2/t^2)} \leq 2\}.$$
\end{definition}

\begin{lemma}[Theorem 4.6.1 \cite{Vershynin18}]\label{lem:sub-gaussian-spectrum}
    Let $A$ be an $n \times m$ matrix, whose columns $A_i$
    are independent, mean zero, sub-Gaussian isotropic
    random vectors in $\R^n$. Then for any $t \geq 0$,
    we have
    $$\sqrt{m} - CK^2(\sqrt{n}+t) \leq s_n(A) \leq
        s_1(A) \leq \sqrt{m} + CK^2(\sqrt{n}+t)$$
    with probability at least $1-2\mathrm{exp}(-t^2)$.
    Here, $K = \max_i\|A\|_{\psi^2}$ (sub-Gaussian norm
    of $A$).
\end{lemma}

In the above, $\|A\|_{\psi^2} \in O(1)$ if the distribution
in question is $\cN(\vec{0},\id)$. The following corollary
generalises the above result for arbitrary Gaussians.

\begin{corollary}\label{coro:normal-spectrum}
    Let $A$ be an $n \times m$ matrix, whose columns $A_i$
    are independent, random vectors in $\R^n$ from $\cN(\vec{0},\Sigma)$.
    Then for any $t \geq 0$, we have
    $$(\sqrt{m} - CK^2(\sqrt{n} + t))\sqrt{s_n(\Sigma)}
        \leq s_n(A) \leq (\sqrt{m} + CK^2(\sqrt{n} + t))\sqrt{s_n(\Sigma)}$$
    and
    $$s_1(A) \leq (\sqrt{m} + CK^2(\sqrt{n} + t))\sqrt{s_1(\Sigma)}$$
    with probability at least $1-2\mathrm{exp}(-t^2)$.
    Here, $K = \max_i\|A\|_{\psi^2}$ (sub-Gaussian norm
    of $A$).
\end{corollary}
\begin{proof}
    First, we prove the lower bound on $s_n(A)$. Note
    that $s_n(A) = \min\limits_{\|x\|>0}{\frac{\|Ax\|}{\|x\|}}$,
    and that the columns of $\Sigma^{-\frac{1}{2}}A$ are distributed
    as $\cN(\vec{0},\id)$. Therefore, we have the following.
    \begin{align*}
        \min\limits_{\|x\|>0}\frac{\|Ax\|}{\|x\|} &=
                \min\limits_{\|x\|>0}\frac{\|\Sigma^{\frac{1}{2}}
                \Sigma^{-\frac{1}{2}} Ax\|}{\|x\|}\\
            &= \min\limits_{\|x\|>0}\frac{\|\Sigma^{\frac{1}{2}}
                \Sigma^{-\frac{1}{2}} Ax\|}{\|\Sigma^{-\frac{1}{2}}Ax\|}
                \frac{\|\Sigma^{-\frac{1}{2}}Ax\|}{\|x\|}\\
            &\geq \min\limits_{\|x\|>0}\frac{\|\Sigma^{\frac{1}{2}}
                \Sigma^{-\frac{1}{2}} Ax\|}{\|\Sigma^{-\frac{1}{2}}Ax\|}
                \min\limits_{\|x\|>0}\frac{\|\Sigma^{-\frac{1}{2}}Ax\|}{\|x\|}\\
            &\geq \min\limits_{\|y\|>0}\frac{\|\Sigma^{\frac{1}{2}}y\|}{\|y\|}
                \min\limits_{\|x\|>0}\frac{\|\Sigma^{-\frac{1}{2}}Ax\|}{\|x\|}\\
            &\geq (\sqrt{m}-CK^2(\sqrt{n}+t))\sqrt{s_n(\Sigma)}
                \tag{Lemma~\ref{lem:sub-gaussian-spectrum}}
    \end{align*}
    Next, we prove the upper bound on $s_n(A)$. For this,
    we first show that for $X\in \RR^{m\times d}$ and $Y \in \RR^{d \times n}$,
    $s_{\min}(XY) \leq s_{\min}(X)\cdot\|Y\|$.
    \begin{align*}
        s_{\min}(XY) &= \min\limits_{\|z\|=1}\|XYz\|\\
            &\leq \min\limits_{\|z\|=1}\|X\|\|Yz\|\\
            &= \|X\|\cdot \min\limits_{\|z\|=1}\|Yz\|\\
            &= \|X\|\cdot s_{\min}(Y)
    \end{align*}
    Now, $s_{\min}(XY) = s_{\min}(Y^T X^T) \leq \|Y\|\cdot s_{\min}(X)$
    by the above reasoning. Using this results, we have the following.
    \begin{align*}
        s_n(A) &= s_n(\Sigma^{1/2}\cdot\Sigma^{-1/2}A)\\
            &\leq s_n(\Sigma^{1/2})\|\Sigma^{-1/2}A\|\\
            &\leq (\sqrt{m} + CK^2(\sqrt{n}+t))\sqrt{s_n(\Sigma)}
                \tag{Lemma~\ref{lem:sub-gaussian-spectrum}}
    \end{align*}
    Now, we show the upper bound on $s_1(A)$. Note that
    $s_1(A) = \|A\|$.
    \begin{align*}
        \|A\| &= \|\Sigma^{\frac{1}{2}}\Sigma^{-\frac{1}{2}}A\|\\
            &\leq \|\Sigma^{\frac{1}{2}}\|\cdot\|\Sigma^{-\frac{1}{2}}A\|\\
            &\leq (\sqrt{m}+CK^2(\sqrt{n}+t))\sqrt{s_1(\Sigma)}
                \tag{Lemma~\ref{lem:sub-gaussian-spectrum}}
    \end{align*}
    This completes the proof.
\end{proof}

Now, we state a concentration inequality for $\chi^2$ random variables.

\begin{lemma}\label{lem:chi-squared}
    Let $X$ be a $\chi^2$ random variable with $k$
    degrees of freedom. Then,
    $$\pr{}{X>k+2\sqrt{kt}+2t} \leq e^{-t}.$$
\end{lemma}

Next, we state the well-known Bernstein's inequality
for sums of independent Bernoulli random variables.

\begin{lemma}[Bernstein's Inequality]\label{lem:chernoff-add}
    Let $X_1,\dots,X_m$ be independent Bernoulli random variables
    taking values in $\zo$. Let $p = \ex{}{X_i}$.
    Then for $m \geq \frac{5p}{2\epsilon^2}\ln(2/\beta)$ and
    $\eps \leq p/4$,
    $$\pr{}{\abs{\frac{1}{m}\sum{X_i}-p} \geq \epsilon}
        \leq 2e^{-\epsilon^2m/2(p+\epsilon)}
        \leq \beta.$$
\end{lemma}

We finally state a result about the norm of a
vector sampled from $\cN(\vec{0},\id)$.

\begin{lemma}\label{lem:gauss-vector-norm}
    Let $X_1,\dots,X_q \sim \cN(\vec{0},\Sigma)$ be vectors in
    $\RR^d$, where $\Sigma$ is the projection of $\id_{d \times d}$
    on to a subspace of $\RR^d$ of rank $k$. Then
    $$\pr{}{\forall i, \|X_i\|^2 \leq k + 2\sqrt{kt} + 2t} \geq 1-qe^{-t}.$$
\end{lemma}
\begin{proof}
    Since $\Sigma$ is of rank $k$, we can directly
    use Lemma~\ref{lem:chi-squared} for a fixed $i \in [q]$,
    and the union bound over all $i \in [q]$ to get
    the required result. This is because for any $i$,
    $\|X_i\|^2$ is a $\chi^2$ random variable with $k$
    degrees of freedom.
\end{proof}

%\begin{lemma}[Lemma A.3 \cite{KamathSSU19}]\label{lem:gauss-vector-norm}
%    Let $X_1,\dots,X_q \sim \cN(\vec{0},\id)$ be vectors in $\R^d$,
%    where $d \geq O(\log(q))$.
%    Then with probability at least $0.95$, $\|X\| \leq \sqrt{3d}$.
%\end{lemma}

\subsection{Privacy Preliminaries}

\begin{definition}[Differential Privacy (DP) \cite{DworkMNS06}]
    \label{def:dp}
    A randomized algorithm $M:\cX^n \rightarrow \cY$
    satisfies $(\eps,\delta)$-differential privacy
    ($(\eps,\delta)$-DP) if for every pair of
    neighboring datasets $X,X' \in \cX^n$
    (i.e., datasets that differ in exactly one entry),
    $$\forall Y \subseteq \cY~~~
        \pr{}{M(X) \in Y} \leq e^{\eps}\cdot
        \pr{}{M(X') \in Y} + \delta.$$
    When $\delta = 0$, we say that $M$ satisfies
    $\eps$-differential privacy or pure differential
    privacy.
\end{definition}
Neighbouring datasets are those that differ by the replacement of one individual's data. In our setting, each individual's data is assumed to correspond to one point in $\cX = \mathbb{R}^d$, so neighbouring means one point is changed arbitrarily.

Throughout the document, we will assume that $\eps$ is
smaller than some absolute constant less than $1$ for
notational convenience, but note that our results still
hold for general $\eps$. Now, this privacy definition is
closed under post-processing.
%and can be composed with graceful degradation of the privacy
%parameters.
\begin{lemma}[Post Processing \cite{DworkMNS06}]\label{lem:post-processing}
    If $M:\cX^n \rightarrow \cY$ is
    $(\eps,\delta)$-DP, and $P:\cY \rightarrow \cZ$
    is any randomized function, then the algorithm
    $P \circ M$ is $(\eps,\delta)$-DP.
\end{lemma}

%\begin{lemma}[Composition of DP~\cite{DworkMNS06, DworkRV10}]\label{lem:composition}
%    If $M$ is an adaptive composition of differentially
%    private algorithms $M_1,\dots,M_T$, then the following
%    all hold:
%    \begin{enumerate}
%        \item If $M_1,\dots,M_T$ are
%            $(\eps_1,\delta_1),\dots,(\eps_T,\delta_T)$-DP
%            then $M$ is $(\eps,\delta)$-DP for
%            $$\eps = \sum_t \eps_t~~~~\textrm{and}~~~~\delta = \sum_t \delta_t.$$
%        \item If $M_1,\dots,M_T$ are
%            $(\eps_0,\delta_1),\dots,(\eps_0,\delta_T)$-DP
%            for some $\eps_0 \leq 1$, then for every $\delta_0 > 0$, $M$
%            is $(\eps, \delta)$-DP for
%            $$\eps = \eps_0 \sqrt{6 T \log(1/\delta_0)}~~~~
%                \textrm{and}~~~~\delta = \delta_0 + \sum_t \delta_t$$
%    \end{enumerate}
%\end{lemma}
%
%Quantitatively tighter composition results are known \cite[etc.]{KairouzOV15,BunS16}. However, these results suffice for our purposes, as we do not aim to optimize constants.

\subsection{Basic Differentially Private Mechanisms.}
We first state standard results on achieving
privacy via noise addition proportional to
sensitivity~\cite{DworkMNS06}.

\begin{definition}[Sensitivity]
    Let $f : \cX^n \to \R^d$ be a function,
    its \emph{$\ell_1$-sensitivity} and
    \emph{$\ell_2$-sensitivity} are
    $$\Delta_{f,1} = \max_{X \sim X' \in \cX^n} \| f(X) - f(X') \|_1
    ~~~~\textrm{and}~~~~\Delta_{f,2} = \max_{X \sim X' \in \cX^n} \| f(X) - f(X') \|_2,$$
    respectively.
    Here, $X \sim X'$ denotes that $X$ and $X'$
    are neighboring datasets (i.e., those that
    differ in exactly one entry).
\end{definition}

%For functions with bounded $\ell_1$-sensitivity,
%we can achieve $\eps$-DP by adding noise from
%a Laplace distribution proportional to
%$\ell_1$-sensitivity.
%For functions taking values
%in $\R^d$ for large $d$ it is useful to add
%noise from a Gaussian distribution proportional
%to the $\ell_2$-sensitivity, to get $(\eps,\delta)$-DP
%and $\rho$-zCDP.

%\begin{lemma}[Laplace Mechanism] \label{lem:laplacedp}
%    Let $f : \cX^n \to \R^d$ be a function
%    with $\ell_1$-sensitivity $\Delta_{f,1}$.
%    Then the Laplace mechanism
%    $$M(X) = f(X) + \Lap\left(\frac{\Delta_{f,1}}
%        {\eps}\right)^{\otimes d}$$
%    satisfies $\eps$-DP.
%\end{lemma}

%\begin{lemma}[Gaussian Mechanism] \label{lem:gaussiandp}
%    Let $f : \cX^n \to \R^d$ be a function
%    with $\ell_2$-sensitivity $\Delta_{f,2}$.
%    Then the Gaussian mechanism
%    $$M(X) = f(X) + \cN\left(0,\left(\frac{\Delta_{f,2}
%        \sqrt{2\ln(2/\delta)}}{\eps}\right)^2 \cdot \id_{d \times d}\right)$$
%    satisfies $(\eps,\delta)$-DP.
%\end{lemma}

One way of introducing $(\eps,\delta)$-differential
privacy is via adding noise sampled from the truncated
Laplace distribution, proportional to the $\ell_1$
sensitivity.

\begin{lemma}[Truncated Laplace Mechanism \cite{GengDGK20}]\label{lem:truncated-laplace}
    Define the probability density function ($p$) of
    the truncated Laplace distribution as follows.
    \[
    p(x) =
    \begin{cases}
        Be^{-\frac{\abs{x}}{\lambda}} & \text{if $x \in [-A,A]$}\\
        0 & \text{otherwise}
    \end{cases}
    \]
    In the above,
    $$
        \lambda = \frac{\Delta}{\eps},~~~
        A = \frac{\Delta}{\eps}\log\left(1+\frac{e^\eps - 1}{2\delta}\right),~~~
        B = \frac{1}{2\lambda(1-e^{-\frac{A}{\lambda}})}.
    $$
    Let $\TLap(\Delta,\eps,\delta)$ denote a draw
    from the above distribution. 
    
    Let $f : \cX^n \to \R^d$ be a function
    with sensitivity $\Delta$. Then the truncated
    Laplace mechanism
    $$M(X) = f(X) + \TLap(\Delta,\eps,\delta)$$
    satisfies $(\eps,\delta)$-DP.
\end{lemma}

In the above $A \leq \tfrac{\Delta_{f,1}}{\eps}\log(1/\delta)$
since $\eps$ is smaller than some absolute constant less than
$1$. Now, we introduce differentially private histograms.

\begin{lemma}[Private Histograms]\label{lem:priv-hist}
    Let $n \in \mathbb{N}$, $\varepsilon,\delta,\beta>0$, and $\mathcal{X}$ a set.
    There exists $M : \mathcal{X}^n \to \mathbb{R}^{\mathcal{X}}$ which is $(\varepsilon,\delta)$-differentially private and, for all $x \in \mathcal{X}^n$, we have \[\pr{M}{\sup_{y \in \mathcal{X}} \left|M(x)_y - \frac1n |\{ i \in [n] : x_i = y\}| \right| \le O\left(\frac{\log(1/\delta\beta)}{\varepsilon n}\right) } \ge 1-\beta.\]
%    Furthermore, the runtime of $M$ is $\poly(n,\log(\abs{\mathcal{X}}/\eps\beta))$
%    Let $(X_1,\dots,X_n)$ be samples in some data universe
%    $U$, and let $\Omega = \{h_u\}_{u \subset U}$
%    be a collection of disjoint histogram buckets over $U$.
%    Then we have $(\eps,\delta)$-DP
%    histogram algorithms with the following guarantees.
%    \begin{itemize}
%        \item $(\eps,\delta)$-DP:
%            $\ell_\infty$ error - $O\left(\tfrac{\log(1/\delta\beta)}{\eps}\right)$
%            with probability at least $1-\beta$;
%            run time - $\poly(n,\log(\abs{U}/\eps\beta))$
%    \end{itemize}
\end{lemma}

The above holds due
to \cite{BunNS16,Vadhan17}. Finally, we introduce the
$\GAPMAX$ algorithm from \cite{BunDRS18} that outputs
the element from the output space that has the highest
score function, given that there is a significant gap
between the scores of the highest and the second to the
highest elements.

\begin{lemma}[$\GAPMAX$ Algorithm \cite{BunDRS18}]\label{lem:gap-max}
    Let $\Score:\cX^n \times \cY \rightarrow \RR$ be a
    score function with sensitivity $1$ in its
    first argument, and let $\varepsilon,\delta>0$. Then there exists a $(\varepsilon,\delta)$-differentially private algorithm $M : \mathcal{X}^n \to \mathcal{Y}$ and $\alpha=\Theta(\log(1/\delta)/\varepsilon n)$ with the following property. Fix an input
    $X \in \cX^n$. Let
    $$y^* = \argmax_{y \in \cY}\{\Score(X,y)\}.$$
    Suppose
    $$\forall y \in \cY, y \neq y^* \implies \Score(X,y) < \Score(X,y^*)-\alpha n.$$
    Then $M$ outputs $y^*$ with probability 1.
\end{lemma}

%removed preliminaries
%\medskip
%We provide definitions and background results on
%linear algebra and privacy in Appendix \ref{sec:preliminaries}.

\section{Exact case}

Here, we discuss the case, where all $n$ points lie \emph{exactly} in a subspace $s_*$ of dimension $k$ of $\RR^d$. Our goal
is to privately output that subspace. We do it under the
assumption that all strict subspaces of $s_*$ contain at most $\ell$
points.
If the points are in general position, then $\ell=k-1$, as any strictly smaller subspace has dimension $<k$ and cannot contain more points than its dimension.
Let $\mathcal{S}_d^k$ be the set of all $k$-dimensional
subspaces of $\mathbb{R}^d$. Let $\mathcal{S}_d$ be the
set of all subspaces of $\mathbb{R}^d$. We formally define
that problem as follows.

\begin{problem}\label{prob:exact}
    Assume (i) all but at most $\ell$, input points are in some
    $s_* \in \mathcal{S}_d^k$, and (ii)  every subspace
    of dimension $<k$ contains at most $\ell$ points. (If the points
    are in general position -- aside from being contained in
    $s_*$ -- then $\ell=k-1$.) The goal is to output a representation
    of $s_*$.
\end{problem}

We call these $\leq \ell$ points that do not lie in
$s_*$, ``adversarial points''.
With the problem defined in Problem~\ref{prob:exact}, we
will state the main theorem of this section.

\begin{theorem}\label{thm:exact}
    For any $\eps,\delta>0$, $\ell \ge k-1 \ge 0$, and
    $$n \geq O\left(\ell + \frac{\log(1/\delta)}{\eps}\right),$$ there
    exists an $(\eps,\delta)$-DP algorithm
    $M : \mathbb{R}^{d \times n} \to \mathcal{S}_d^k$, such that if
    $X$ is a dataset of $n$ points satisfying the conditions
    in Problem~\ref{prob:exact},
    then $M(X)$ outputs a representation of $s_*$ with probability $1$.
\end{theorem}

We prove Theorem \ref{thm:exact} by proving the privacy and
the accuracy guarantees of Algorithm~\ref{alg:exact}.
The algorithm performs a $\GAPMAX$ (cf.~Lemma~\ref{lem:gap-max}).
It assigns a score to all the relevant
subspaces, that is, the subspaces spanned by the points
of the dataset $X$. We show that the only subspace that
has a high score is the true subspace $s_*$, and the rest
of the subspaces have low scores. Then $\GAPMAX$ outputs
the true subspace successfully because of the gap between
the scores of the best subspace and the second to the best
one. For $\GAPMAX$ to work all the time, we define a default
option in the output space that has a high score, which we
call $\NULL$. Thus, the output space is now
$\cY = \cS_d \cup \{\NULL\}$. Also, for $\GAPMAX$ to run in
finite time, we filter $\cS_d$ to select finite number of subspaces
that have at least $0$ scores on the basis of $X$. Note that
this is a preprocessing step, and does not violate privacy as,
we will show, all other subspaces already have $0$ probability
of getting output.
We define the score function
$u : \mathcal{X}^n \times \cY \to \mathbb{N}$
as follows.
\[ u(x,s) :=
\begin{cases}
    |x \cap s| - \sup \{ |x \cap t| : t \in \mathcal{S}_d, t \subsetneq s \} &
        \text{if $s \in \cS_d$}\\
    \ell + \frac{4\log(1/\delta)}{\eps} + 1 & \text{if $s = \NULL$}
\end{cases}
\]
Note that this score function can be computed in finite
time because for any $m$ points and $i>0$, if the points
are contained in an $i$-dimensional subspace, then the
subspace that contains all $m$ points must lie within
the set of subspaces spanned by ${m \choose i+1}$ subsets
of points.

\begin{algorithm}[h!] \label{alg:exact}
\caption{DP Exact Subspace Estimator
    $\DPESE_{\eps, \delta, k, \ell}(X)$}
\KwIn{Samples $X \in \R^{d \times n}$.
    Parameters $\eps, \delta, k, \ell > 0$.}
\KwOut{$\hat{s} \in \cS_d^k$.}
\vspace{5pt}

Set $\cY \gets \{\NULL\}$ and sample noise $\xi(\NULL)$ from $\TLap(2,\eps,\delta)$.\\
Set score $u(X,\NULL) = \ell + \frac{4\log(1/\delta)}{\eps} + 1$.
\vspace{5pt}

\tcp{Identify candidate outputs.}
\For{each subset $S$ of $X$ of size $k$}{
    Let $s$ be the subspace spanned by $S$.\\
    $\cY \gets \cY \cup \{s\}$.\\
    Sample noise $\xi(s)$ from $\TLap(2,\eps,\delta)$.\\
    Set score $u(X,s) = |x \cap s| - \sup \{ |x \cap t| : t \in \mathcal{S}_d, t \subsetneq s \} $.
}
\vspace{5pt}

\tcp{Apply $\GAPMAX$.}
Let $s_1 = \argmax_{s \in \cY} u(X,s)$ be the candidate with the largest score.\\
Let $s_2 = \argmax_{s \in \cY \setminus \{s_1\}} u(X,s)$ be the candidate with the second-largest score.\\
Let $\hat s = \argmax_{s \in \cY} \max\{ 0 , u(X,s) - u(X,s_2) -1\} + \xi(s)$.\\
\tcp{Truncated Laplace noise $\xi \sim \TLap(2,\eps,\delta)$; see Lemma \ref{lem:truncated-laplace}}

\vspace{5pt}
\Return $\hat{s}.$
\vspace{5pt}
\end{algorithm}

We split the proof of Theorem~\ref{thm:intro-main-exact} into sections
for privacy (Lemma~\ref{lem:exact-privacy}) and accuracy (Lemma~\ref{lem:exact-accuracy}).

\subsection{Privacy}

\begin{lemma}\label{lem:exact-privacy}
    Algorithm~\ref{alg:exact} is $(\eps,\delta)$-differentially
    private.
\end{lemma}
The proof of Lemma \ref{lem:exact-privacy} closely follows the
privacy analysis of $\GAPMAX$ by \cite{BunDRS18}. The only novelty
is that Algorithm \ref{alg:exact} may output $\NULL$ in the case
that the input is malformed (i.e., doesn't satisfy the assumptions
of Problem~\ref{prob:exact}).

The key is that the score $u(X,s)$ is low sensitivity. Thus
$\max\{ 0 , u(X,s) - u(X,s_2) -1\}$ also has low sensitivity.
What we gain from subtracting the second-largest score and
taking this maximum is that these values are also sparse -- only
one ($s=s_1$) is nonzero. This means we can add noise to all
the values without paying for composition. We now prove
Lemma~\ref{lem:exact-privacy}.

\begin{proof}
    First, we argue that the sensitivity of $u$ is
    $1$. %Consider any $s \in \cS_d$, and neighbouring datasets $X,X'$.
    The quantity $\abs{X \cap s}$ has sensitivity $1$ and so does
    $\sup \{ |X \cap t| : t \in \mathcal{S}_d, t \subsetneq s \}$.
    This implies sensitivity $2$ by the triangle inequality.
    However, we see that it is not possible to change one point
    that simultaneously increases $\abs{X \cap s}$ and decreases
    $\sup \{ |X \cap t| : t \in \mathcal{S}_d, t \subsetneq s \}$
    or vice versa. Thus the sensitivity is actually $1$.

    We also argue that $u(X,s_2)$ has sensitivity $1$, where $s_2$
    is the candidate with the second-largest score.
    Observe that the second-largest score is a monotone function of
    the collection of all scores -- i.e., increasing scores cannot
    decrease the second-largest score and vice versa.
    Changing one input point can at most increase all the scores by
    $1$, which would only increase the second-largest score by $1$.

    This implies that $\max\{ 0, u(X,s) - u(X,s_2) -1 \}$ has sensitivity
    $2$ by the triangle inequality and the fact that the maximum does
    not increase the sensitivity.

    Now we observe that for any input $X$ there is at most one $s$ such
    that $\max\{ 0, u(X,s) - u(X,s_2) -1 \} \ne 0$, namely $s=s_1$.
    We can say something even stronger: Let $X$ and $X'$ be neighbouring
    datasets with $s_1$ and $s_2$ the largest and second-largest scores
    on $X$ and $s_1'$ and $s_2'$ the largest and second-largest scores
    on $X'$. Then there is at most one $s$ such that
    $\max\{ 0, u(X,s) - u(X,s_2) -1 \} \ne 0$ or $\max\{ 0, u(X',s) - u(X',s_2') -1 \} \ne 0$.
    In other words, we cannot have both $u(X,s_1) - u(X,s_2) >1$ and
    $u(X',s_1') - u(X',s_2') >1$ unless $s_1=s_1'$. This holds because
    $u(X,s) - u(X,s_2)$ has sensitivity $2$.

    With these observations in hand, we can delve into the privacy
    analysis. Let $X$ and $X'$ be neighbouring datasets with $s_1$
    and $s_2$ the largest and second-largest scores on $X$ and $s_1'$
    and $s_2'$ the largest and second-largest scores on $X'$. Let $\cY$
    be the set of candidates from $X$ and let $\cY'$ be the set of
    candidates from $X'$. Let $\check \cY = \cY \cup \cY'$ and
    $\hat \cY = \cY \cap \cY'$.

    We note that, for $s \in \check \cY$, if $u(X,s) \le \ell$, then
    there is no way that $\hat s = s$. This is because
    $|\xi(s)|\le \frac{2 \log(1/\delta)}{\varepsilon}$ for all $s$ and
    hence, there is no way we could have
    $\argmax_{s \in \cY} \max\{ 0 , u(X,s) - u(X,s_2) -1\} +
    \xi(s) \ge \argmax_{s \in \cY} \max\{ 0 , u(X,\NULL) - u(X,s_2) -1\} + \xi(\NULL)$.

    If $s \in \check \cY \setminus \hat \cY$, then
    $u(X,s) \le |X \cap s| \le k+1 \le \ell$ and $u(X',s) \le \ell$.
    This is because $s \notin \hat \cY$ implies $|X \cap s| < k$ or
    $|X' \cap s| < k$, but $|X \cap s| \le |X' \cap s| +1$. Thus,
    there is no way these points are output and, hence, we can ignore
    these points in the privacy analysis. (This is the reason for adding
    the $\NULL$ candidate.)

    Now we argue that the entire collection of noisy values
    $\max\{ 0 , u(X,s) - u(X,s_2) -1\} + \xi(s)$ for $s \in \hat \cY$
    is differentially private.
    This is because we are adding noise to a vector where (i) on the
    neighbouring datasets only $1$ coordinate is potentially different
    and (ii) this coordinate has sensitivity $2$.
\end{proof}

\subsection{Accuracy}

We start by showing that the true subspace $s_*$ has a
high score, while the rest of the subspaces have low scores.

\begin{lemma}\label{lem:scores}
    Under the assumptions of Problem~\ref{prob:exact}, $u(x,s_*) \ge n - 2\ell$
    and $u(x,s')\le 2\ell$ for $s' \ne s_*$.
\end{lemma}
\begin{proof}
    We have $u(x,s_*) = |x \cap s_*| - |x \cap s'|$ for some
    $s' \in \mathcal{S}_d$ with $s' \subsetneq s_*$. The
    dimension of $s'$ is at most $k-1$ and, by the assumption
    (ii), $|x \cap s'| \le \ell$. 
    
    Let $s' \in \mathcal{S}_d \setminus \{s_*\}$.
    There are three cases to analyse:
    \begin{enumerate}
        \item Let $s' \supsetneq s_*$. Then
            $u(x,s') \le |x \cap s'| - |x \cap s_*| \leq \ell$
            because the $\leq \ell$ adverserial points and the $\geq n-\ell$
            non-adversarial points may not together lie in a subspace
            of dimension $k$.

        \item Let $s' \subsetneq s_*$. Let
            $k'$ be the dimension of $s'$. Clearly $k'<k$.
            By our assumption (ii), $|s' \cap x| \le \ell$.
            Then $u(x,s') = |x \cap s'| - |x \cap t| \le \ell$
            for some $t$ because the $\leq \ell$ adversarial points
            already don't lie in $s_*$, so they will not lie in
            any subspace of $s_*$.

        \item Let $s'$ be incomparable to $s_*$.
            Let $s'' = s' \cap s_*$. Then
            $u(x,s') \le |x \cap s'| - |x \cap s''| \leq \ell$
            because the adversarial points may not lie in $s_*$,
            but could be in $s'\setminus s''$.
    \end{enumerate}
    This completes the proof.
\end{proof}

Now, we show that the algorithm is accurate.

\begin{lemma}\label{lem:exact-accuracy}
    If
    $n \geq 3\ell + \frac{8\log(1/\delta)}{\eps} + 2,$
    then Algorithm~\ref{alg:exact} outputs $s_*$ for Problem~\ref{prob:exact}.
\end{lemma}
\begin{proof}
    From Lemma~\ref{lem:scores}, we know that $s_*$ has
    a score of at least $n-2\ell$, and the next best subspace
    can have a score of at most $\ell$. Also, the score of
    $\NULL$ is defined to be $\ell + \tfrac{4\log(1/\delta)}{\eps} + 1$.
    This means that the gap satisfies $\max\{ 0 , u(X,s_*) - u(X,s_2) -1\} \ge n - 3\ell - \tfrac{4\log(1/\delta)}{\eps}  -1$.
    Since the noise is bounded by $\tfrac{2\log(1/\delta)}{\eps}$, our bound on $n$ implies that $\hat s = s_*$
\end{proof}

%\subsection{Putting It All Together}
%
%\begin{proof}[Proof of Theorem~\ref{thm:exact}]
%    By Lemmata~\ref{lem:exact-privacy} and \ref{lem:exact-accuracy},
%    the algorithm is differentially private and accurate.
%\end{proof}

\begin{comment}
    
    Key questions: Remove/weaken assumption (ii) general
    position. Approximate case (or discretization).
    
    \subsection{Generalizing beyond general position}
    
    Still assuming $x \subset s_*$ for some $s_* \in \mathcal{S}_d^k$. Also assume $0 \notin x$.
    
    Claim: There exists some $s \in \mathcal{S}_d$ with $u(x,s) \ge n/k$.
    \begin{proof}
    Let $s_k=s_*$. For $i=k,k-1,k-2,\cdots,1$, inductively choose $s_{i-1} \in \mathcal{S}_d$ such that $s_{i-1} \subsetneq s_i$ and $|x \cap s_{i-1}|$ is maximal. By construction, $u(x,s_i) = |x \cap s_i| - |x \cap s_{i-1}|$ for all $i \in [k]$. By assumption, $|x \cap s_k| = n$. Also, $\mathsf{dimension}(s_i) \le i$. Thus $|x \cap s_0| = 0$.
    Since $\sum_{i=1}^k u(x,s_i) = |x \cap s_k| - |x \cap s_0| = n$, we must have $u(x,s_i) \ge n/k$ for some $i \in [k]$.
    \end{proof}

\end{comment}

\subsection{Lower Bound}

Here, we show that our upper bound is optimal up to constants for
the exact case.

\begin{theorem}
     Any $(\eps,\delta)$-DP algorithm that takes a dataset of $n$ points satisfying the conditions
    in Problem~\ref{prob:exact} and outputs $s_*$ with probability $>0.5$ requires
    $n \geq \Omega\left(\ell + \frac{\log(1/\delta)}{\eps}\right).$
\end{theorem}
\begin{proof}
    First, $n \ge \ell + k$. This is because we need at least
    $k$ points to span the subspace, and $\ell$ points could be corrupted.
    Second, $n \ge \Omega(\log(1/\delta)/\varepsilon)$ by group
    privacy. Otherwise, the algorithm is $(10,0.1)$-differentially
    private with respect to changing the \emph{entire} dataset and
    it is clearly impossible to output the subspace under this condition.
\end{proof}

\section{Approximate Case}

In this section, we discuss the case, where the data
``approximately'' lies in a $k$-dimensional subspace of
$\R^d$. %An alternate perspective of this problem is that
%the data lies in a $k$-dimensional subspace of $\R^d$,
%but has very small noise in the orthogonal directions.
We make a Gaussian distributional assumption, where the
covariance is approximately $k$-dimensional, though the
results could be extended to distributions with heavier
tails using the right inequalities. We formally define
the problem:

\begin{problem}\label{prob:gaussians}
    Let $\Sigma \in \R^{d \times d}$ be a symmetric, PSD
    matrix of rank $\geq k \in \{1,\dots,d\}$, and let $0 < \gamma \ll 1$,
    such that $\tfrac{\lambda_{k+1}}{\lambda_k} \leq \gamma^2$.
    Suppose $\Pi$ is the projection matrix corresponding
    to the subspace spanned by the eigenvectors of $\Sigma$
    corresponding to the eigenvalues $\lambda_1,\dots,\lambda_k$.
    Given sample access to $\cN(\vec{0},\Sigma)$,
    and $0 < \alpha < 1$, output a projection matrix $\wh{\Pi}$,
    such that $\|\Pi-\wh{\Pi}\| \leq \alpha$.
\end{problem}

We solve Problem~\ref{prob:gaussians} under the constraint
of $(\eps,\delta)$-differential privacy. Throughout this section,
we would refer to the subspace spanned by the top $k$ eigenvectors
of $\Sigma$ as the ``true'' or ``actual'' subspace.

Algorithm \ref{alg:approximate} solves Problem~\ref{prob:gaussians} and proves Theorem \ref{thm:intro-main-approx}.
Here $\|\cdot\|$ is the operator norm.

\begin{remark}\label{rem:gamma}
    We scale the eigenvalues of $\Sigma$
    so that $\lambda_k=1$ and $\lambda_{k+1} \leq \gamma^2$.
    %We will be adopting this notation throughout this text.
    Also, for the purpose of the analysis, we will be splitting
    $\Sigma = \Sigma_k + \Sigma_{d-k}$, where $\Sigma_k$ is the
    covariance matrix formed by the top $k$ eigenvalues and
    the corresponding eigenvectors of $\Sigma$ and $\Sigma_{d-k}$
    is remainder.
\end{remark}

%\begin{comment}

Also, we assume the
knowledge of $\gamma$ (or an upper bound on $\gamma$). Our solution
is presented in Algorithm~\ref{alg:approximate}. The following
theorem is the main result of the section.

\begin{theorem}\label{thm:approximate}
    Let $\Sigma \in \R^{d \times d}$ be an arbitrary, symmetric, PSD
    matrix of rank $\geq k \in \{1,\dots,d\}$, and let $0 < \gamma < 1$.
    Suppose $\Pi$ is the projection matrix corresponding
    to the subspace spanned by the vectors of $\Sigma_k$.
    Then given
    $$\gamma^2 \in
        O\left(\frac{\eps\alpha^2n}{d^{2}k\ln(1/\delta)}\cdot
        \min\left\{\frac{1}{k},
        \frac{1}{\ln(k\ln(1/\delta)/\eps)}
        \right\}\right),$$
    such that $\lambda_{k+1}(\Sigma) \leq \gamma^2\lambda_k(\Sigma)$,
    for every $\eps,\delta>0$, and $0 < \alpha < 1$,
    there exists and $(\eps,\delta)$-DP algorithm that takes
    $$n \geq O\left(\frac{k\log(1/\delta)}{\eps} +
        \frac{\log(1/\delta)\log(\log(1/\delta)/\eps)}{\eps}\right)$$
    samples from $\cN(\vec{0},\Sigma)$, and outputs a projection matrix $\wh{\Pi}$,
    such that $\|\Pi-\wh{\Pi}\| \leq \alpha$ with probability
    at least $0.7$.
    \tnote{Eigenvalue gap assumption missing. Also order of quantifiers is ambiguous -- ``for every $\Sigma$ there exists an algorith.''}
\end{theorem}
%\end{comment}

Algorithm~\ref{alg:approximate} is a type of
``Subsample-and-Aggregate'' algorithm \cite{NissimRS07}.
Here, we consider multiple subspaces formed by the points
from the same Gaussian, and privately find a subspace that
is close to all those subspaces. Since the subspaces formed
by the points would be close to the true subspace, the privately
found subspace would be close to the true subspace.

A little more formally, we first sample $q$ public data points
(called ``reference points'') from $\cN(\vec{0},\id)$. Next,
we divide the original dataset $X$ into disjoint datasets of $m$ samples
each, and project all reference points on the subspaces spanned
by every subset. Now, for every reference point, we do the
following. We have $t=\tfrac{n}{m}$ projections of the reference
point. Using DP histogram over $\R^d$, we aggregate those
projections in the histogram cells; with high probability
all those projections will be close to one another, so they
would lie within one histogram cell. We output a random point
from the histogram cell corresponding to the reference point.
With a total of $q$ points output in this way, we finally
output the projection matrix spanned by these points. In
the algorithm $C_0$, $C_1$, and $C_2$ are universal constants.

We divide the proof of Theorem~\ref{thm:approximate}
into two parts: privacy (Lemma \ref{coro:privacy}) and
accuracy (Lemma~\ref{lem:final-projection}).

\begin{algorithm}[h!] 
\caption{\label{alg:approximate}DP Approximate Subspace Estimator
    $\DPASE_{\eps, \delta, \alpha, \gamma, k}(X)$}
\KwIn{Samples $X_1,\dots,X_{n} \in \R^d$.
    Parameters $\eps, \delta, \alpha, \gamma, k > 0$.}
\KwOut{Projection matrix $\wh{\Pi} \in \R^{d \times d}$ of rank $k$.}
\vspace{5pt}

Set parameters:
    $t \gets \tfrac{C_0\ln(1/\delta)}{\eps}$ \qquad
    $m \gets \lfloor n/t \rfloor$ \qquad $q \gets C_1 k$
    \qquad $\ell \gets \tfrac{C_2\gamma\sqrt{dk}(\sqrt{k}+\sqrt{\ln(kt)})}{\sqrt{m}}$
\vspace{5pt}

Sample reference points $p_1,\dots,p_q$ from $\cN(\vec{0},\id)$ independently.
\vspace{5pt}

\tcp{Subsample from $X$, and form projection matrices.}
\For{$j \in 1,\dots,t$}{
    Let $X^j = (X_{(j-1)m+1},\dots,X_{jm}) \in \mathbb{R}^{d \times m}$.\\
    Let $\Pi_j \in \mathbb{R}^{d \times d}$ be the projection matrix onto the subspace spanned by the eigenvectors of $X^j (X^j)^T \in \mathbb{R}^{d \times d}$ corresponding to the largest $k$ eigenvalues.\\
    \For{$i \in 1,\dots,q$}{
        $p_{i}^j \gets \Pi_j p_i$
    }
}
\vspace{5pt}

\tcp{Create histogram cells with random offset.}
Let $\lambda$ be a random number in $[0,1)$.\\
Divide $\R^{qd}$ into $\Omega =
    \{\dots,[\lambda\ell+i\ell,\lambda\ell+(i+1)\ell),\dots\}^{qd}$,
    for all $i \in \Z$.\\
Let each disjoint cell of length $\ell$ be a histogram bucket.
\vspace{5pt}

\tcp{Perform private aggregation of subspaces.}
For each $i \in [q]$, let $Q_i \in \RR^{d \times t}$ be the
    dataset, where column $j$ is $p_i^j$.\\
Let $Q \in \RR^{qd \times t}$ be the vertical concatenation
    of all $Q_i$'s in order.\\
Run $(\eps,\delta)$-DP histogram over $\Omega$ using $Q$
    to get $\omega \in \Omega$ that contains at least $\tfrac{t}{2}$ points.\\
\If{no such $\omega$ exists}{
    \Return $\bot$
}
\vspace{5pt}

\tcp{Return the subspace.}
Let $\wh{p}=(\wh{p}_1,\dots,\wh{p}_d,\dots,\wh{p}_{(q-1)d+1},\dots,\wh{p}_{qd})$
    be a random point in $\omega$.\\
\For{each $i \in [q]$}{
    Let $\wh{p}_i = (\wh{p}_{(i-1)d+1},\dots,\wh{p}_{id})$.
}
    Let $\wh{\Pi}$ be the projection matrix of the top-$k$ subspace of $(\wh{p}_1,\dots,\wh{p}_q)$.\\
\Return $\wh{\Pi}.$
\vspace{5pt}
\end{algorithm}

\subsection{Privacy}

We analyse the privacy by understanding the sensitivities
at the only sequence of steps invoking a differentially
private mechanism, that is, the sequence of steps involving
DP-histograms.

\begin{lemma}\label{lem:histogram-sensitivity}\label{coro:privacy}
    Algorithm~\ref{alg:approximate} is $(\eps,\delta)$-differentially
    private.
\end{lemma}
\begin{proof}
    Changing one point in $X$ can change only
    one of the $X^j$'s. This can
    only change one point in $Q$, which in turn can only
    change the counts in two histogram cells by $1$.
    Therefore, the sensitivity is $2$. % Since the choice
    %of $i$ was arbitrary, this is true for all $i$.
    %For a reference
    %point $p_i$, changing a point in $X^{j^*}$ can either move
    %its projection on to the subspace spanned by $X^{j^*}$ to a
    %different histogram cell, or keep it in the same cell.
%    Privacy now follows from the guarantees of DP-histogram (Lemma~\ref{lem:priv-hist}).
    Because the sensitivity of the histogram step is bounded
    by $2$ (Lemma~\ref{lem:histogram-sensitivity}), an application
    of DP-histogram, by Lemma~\ref{lem:priv-hist}, is $(\eps,\delta)$-DP.
    Outputting a random
    point in the privately found histogram cell preserves privacy
    by post-processing (Lemma~\ref{lem:post-processing}).
    Hence, the claim.
\end{proof}

\subsection{Accuracy}

Now we delve into the utility analysis of the algorithm.
For $1 \leq j \leq t$,
let $X^j$ be the subsets of $X$ as defined in
Algorithm~\ref{alg:approximate}, and $\Pi_j$ be the
projection matrices of their respective subspaces. We
now show that $\Pi_j$ and the projection matrix of the
subspace spanned by $\Sigma_k$ are close in operator norm.

\begin{lemma}\label{lem:empirical-subspaces-close}
    Let $\Pi$ be the projection matrix of the subspace
    spanned by the vectors of $\Sigma_k$, and for each
    $1 \leq j \leq t$, let $\Pi_j$ be the projection
    matrix as defined in Algorithm~\ref{alg:approximate}.
    If $m \geq O(k + \ln(qt))$, then
    $$\pr{}{\forall j, \|\Pi-\Pi_j\| \leq
        O\left(\frac{\gamma\sqrt{d}}{\sqrt{m}}\right)} \geq 0.95$$
\end{lemma}
\begin{proof}
    We show that the subspaces spanned by $X^j$ and
    the true subspace spanned by $\Sigma$ are close.
    Formally, we invoke
    Lemmata \ref{lem:sin-theta} and \ref{lem:sin-theta-property}.
    This closeness follows from standard matrix concentration
    inequalities.
    %which we discuss in Appendix \ref{sec:preliminaries}.
    
    Fix a $j \in [t]$. Note that $X^j$ can be written
    as $Y^j + H$, where $Y^j$ is the matrix of vectors
    distributed as $\cN(\vec{0},\Sigma_k)$, and $H$ is
    a matrix of vectors distributed as $\cN(\vec{0},\Sigma_{d-k})$,
    where $\Sigma_k$ and $\Sigma_{d-k}$ are defined as
    in Remark~\ref{rem:gamma}.
    By Corollary~\ref{coro:normal-spectrum}, with probability at least $1-\tfrac{0.02}{t}$,
    $s_k(Y^j) \in \Theta((\sqrt{m}+\sqrt{k})(\sqrt{s_k(\Sigma_k)})) = \Theta(\sqrt{m}+\sqrt{k})> 0$.
    Therefore, the subspace spanned by
    $Y^j$ is the same as the subspace spanned by $\Sigma_k$.
    So, it suffices to look at the subspace spanned
    by $Y^j$.

    Now, by Corollary~\ref{coro:normal-spectrum}, we know
    that with probability at least $1-\tfrac{0.02}{t}$,
    $\|X^j-Y^j\| = \|H\| \leq O((\sqrt{m}+{\sqrt{d}})\sqrt{s_1(\Sigma_{d-k})})
    \leq O(\gamma(\sqrt{m}+\sqrt{d})\sqrt{s_k(\Sigma_k)}) \leq O(\gamma(\sqrt{m}+\sqrt{d}))$.
    
    We wish to invoke Lemma~\ref{lem:sin-theta}. Let $UDV^T$
    be the SVD of $Y^j$, and let $\hat{U}\hat{D}\hat{V}^T$ be
    the SVD of $X^j$. Now, for a matrix $M$, let $\Pi_M$ denote
    the projection matrix of the subspace spanned by the columns
    of $M$. Define quantities $a,b,z_{12},z_{21}$ as follows.
    \begin{align*}
        a &= s_{\min}(U^TX^jV)\\
            &= s_{\min}(U^TY^jV + U^THV)\\
            &= s_{\min}(U^TY^jV) \tag{Columns of $U$ are orthogonal to columns of $H$}\\
            &= s_k(Y^j)\\
            &\in \Theta(\sqrt{m}+\sqrt{k})\\
            &\in \Theta(\sqrt{m})\\
        b &= \|U_{\bot}^TX^jV_{\bot}\|\\
            &= \|U_{\bot}^TY^jV_{\bot} + U_{\bot}^THV_{\bot}\|\\
            &= \|U_{\bot}^THV_{\bot}\|
                \tag{Columns of $U_{\bot}$ are orthogonal to columns of $Y^j$}\\
            &\leq \|H\|\\
            &\leq O(\gamma(\sqrt{m}+\sqrt{d}))\\
        z_{12} &= \|\Pi_U H \Pi_{V_{\bot}}\|\\
            &= 0\\
        z_{21} &= \|\Pi_{U_{\bot}}H\Pi_V\|\\
            &= \|\Pi_{U_{\bot}}\Sigma_{d-k}^{1/2}(\Sigma_{d-k}^{-1/2}H)\Pi_V\|
    \end{align*}
    Now, in the above, $\Sigma_{d-k}^{-1/2}H \in \RR^{d\times m}$,
    such that each of its entry is an independent sample from $\cN(0,1)$.
    Right-multiplying it by $\Pi_V$ makes it a matrix
    in a $k$-dimensional subspace of $\RR^m$, such that
    each row is an independent vector from a spherical
    Gaussian. Using Corollary~\ref{coro:normal-spectrum},
    $\|\Sigma_{d-k}^{-1/2}H\| \leq O(\sqrt{d}+\sqrt{k}) \leq O(\sqrt{d})$
    with probability at least $1-\tfrac{0.01}{t}$.
    Also, $\|\Pi_{U_{\bot}}\Sigma_{d-k}^{1/2}\| \leq O(\gamma\sqrt{s_k(\Sigma_k)}) \leq O(\gamma)$.
    This gives us:
    $$z_{21} \leq O(\gamma\sqrt{d}).$$

    Since $a^2 > 2b^2$, we get the following by
    Lemma~\ref{lem:sin-theta}.
    \begin{align*}
        \|\text{Sin}(\Theta)(U,\hat{U})\| &\leq \frac{az_{21} + bz_{12}}
                {a^2-b^2-\min\{z_{12}^2,z_{21}^2\}}\\
            &\leq O\left(\frac{\gamma\sqrt{d}}{\sqrt{m}}\right)
    \end{align*}

    Therefore, using Lemma~\ref{lem:sin-theta-property},
    and applying the union bound over all $j$, we get the
    required result.
\end{proof}

Let $\xi = O\left(\tfrac{\gamma\sqrt{d}}{\sqrt{m}}\right)$. We
show that the projections of any
reference point are close.

\begin{corollary}\label{coro:reference-projections-close}
    Let $p_1,\dots,p_q$ be the reference points as
    defined in Algorithm~\ref{alg:approximate}, and
    let $\Pi$ and $\Pi_j$ (for $1 \leq j \leq t$) be
    projections matrices as defined in Lemma~\ref{lem:empirical-subspaces-close}.
    Then
    $$\pr{}{\forall i,j, \|(\Pi-\Pi_j)p_i\| \leq O(\xi(\sqrt{k}+\sqrt{\ln(qt)}))} \geq 0.9.$$
\end{corollary}
\begin{proof}
    We know from Lemma~\ref{lem:empirical-subspaces-close}
    that $\|\Pi-\Pi_j\| \leq \xi$ for all $j$ with
    probability at least $0.95$. For $j \in [t]$, let
    $\wh{\Pi}_j$ be the projection matrix for the union
    of the $j^{\text{th}}$ subspace and the subspace
    spanned by $\Sigma_k$. Lemma~\ref{lem:gauss-vector-norm}
    implies that with probability at least $0.95$,
    for all $i,j$, $\|\wh{\Pi}_j p_i\| \leq O(\sqrt{k}+\sqrt{\ln(qt)})$.
    Therefore,
    \begin{align*}
        \|(\Pi-\Pi_j)p_i\| &= \|(\Pi-\Pi_j)\wh{\Pi}_jp_i\|
            \leq \|\Pi-\Pi_j\|\cdot\|\wh{\Pi}_jp_i\|
            \leq O(\xi(\sqrt{k}+\sqrt{\ln(qt)})).
    \end{align*}
    Hence, the claim.
\end{proof}

The above corollary shows that the projections of
each reference point lie in a ball of radius $O(\xi\sqrt{k})$.
Next, we show that for each reference point, all the
projections of the point lie inside a histogram cell
with high probability. For notational convenience, since
each point in $Q$ is a concatenation of the projection
of all reference points on a given subspace, for all
$i,j$, we refer to
$(0,\dots,0,Q_{(i-1)d+1}^j,\dots,Q_{id}^j,0,\dots,0) \in R^{qd}$
(where there are $(i-1)d$ zeroes behind $Q_{(i-1)d+1}^j$,
and $(q-i)d$ zeroes after $Q_{id}^j$) as $p_i^j$.

\begin{lemma}\label{lem:histogram-cell-points}
    Let $\ell$ and $\lambda$ be the length of a histogram
    cell and the random offset respectively, as defined in
    Algorithm~\ref{alg:approximate}. Then
    $$\pr{}{|\omega \cap Q| = t} \geq 0.8.$$
    Thus there exists $\omega \in \Omega$ that,
    such that all points in $Q$ lie within $\omega$.
\end{lemma}
\begin{proof}
    Let $r = O(\xi(\sqrt{k}+\sqrt{\ln(qt)}))$. This implies that $\ell = 20r\sqrt{q}$.
    The random offset could also be viewed as moving along a
    diagonal of a cell by $\lambda\ell\sqrt{dq}$. We know that
    with probability at least $0.8$, for each $i$, all projections
    of reference point $p_i$ lie in a ball of radius $r$.
    This means that all the points in $Q$ lie in a ball of
    radius $r\sqrt{q}$. Then
    $$\pr{}{|\omega \cap Q| = t} \leq \pr{}{\frac{1}{20} \geq
        \lambda \vee \lambda \geq \frac{19}{20}} = \frac{1}{10}.$$
    Taking the union bound over all $q$ and the failure
    of the event in Corollary~\ref{coro:reference-projections-close},
    we get the claim.
\end{proof}

Now, we analyse the sample complexity due
to the private algorithm, that is,
DP-histograms.

\begin{lemma}\label{lem:dp-histogram-cost}
    Let $\omega$ be the histogram cell as defined in
    Algorithm~\ref{alg:approximate}. Suppose $\pcount(\omega)$
    is the noisy count of $\omega$ as a result of applying
    the private histogram. If
    $t \geq O\left(\frac{\log(1/\delta)}{\eps}\right),$
    then
    $$\pr{}{\abs{\pcount(\omega)} \geq \frac{t}{2}} \geq 0.75.$$
\end{lemma}
\begin{proof}
    Lemma~\ref{lem:histogram-cell-points} implies that
    with probability at least $0.8$, for each $i$, all
    projections of $p_i$ lie in a histogram cell, that is,
    all points of $Q$ lie in a histogram cell in $\Omega$.
    Because of the error bound in Lemma~\ref{lem:priv-hist}
    and our bound on $t$, we see at least $\tfrac{t}{2}$
    points in that cell with probability at least $1-0.05$.
    Therefore, by taking the union bound, the proof is complete.
\end{proof}

We finally show that the error of the projection matrix
that is output by Algorithm~\ref{alg:approximate} is small.

\begin{lemma}\label{lem:final-projection}
    Let $\wh{\Pi}$ be the projection matrix as defined in
    Algorithm~\ref{alg:approximate}, and $n$ be the total
    number of samples. If
    $$\gamma^2 \in
        O\left(\frac{\eps\alpha^2n}{d^{2}k\ln(1/\delta)}\cdot
        \min\left\{\frac{1}{k},
        \frac{1}{\ln(k\ln(1/\delta)/\eps)}
        \right\}\right),$$
    $n \geq O(\frac{k\log(1/\delta)}{\eps}+\frac{\ln(1/\delta)\ln(\ln(1/\delta)/\eps)}{\eps})$,
    and $q \geq O(k)$
    the with probability at least $0.7$, $\|\wh{\Pi}-\Pi\| \leq \alpha$.
\end{lemma}
\begin{proof}
    For each $i \in [q]$, let $p_i^*$ be the projection
    of $p_i$ on to the subspace spanned by $\Sigma_k$,
    $\wh{p}_i$ be as defined in the algorithm, and $p_i^j$
    be the projection of $p_i$ on to the subspace spanned
    by the $j^{\mathrm{th}}$ subset of $X$. From Lemma~\ref{lem:dp-histogram-cost},
    we know that all $p_i^j$'s are contained in a histogram
    cell of length $\ell$. This implies that $p_i^*$ is also
    contained within the same histogram cell.

    Now, let $P=(p_1^*,\dots,p_q^*)$ and $\wh{P}=(\wh{p}_1,\dots,\wh{p}_q)$.
    Then by above, $\wh{P}=P+E$, where $\|E\|_F \leq 2\ell\sqrt{dq}$. Therefore,
    $\|E\| \leq 2\ell\sqrt{dq}$. Let $E=E_P+E_{\wb{P}}$,
    where $E_P$ is the component of $E$ in the subspace
    spanned by $P$, and $E_{\wb{P}}$ be the orthogonal
    component. Let $P' = P + E_P$. We will be analysing
    $\wh{P}$ with respect to $P'$.

    Now, with probability
    at least $0.95$, $s_k(P) \in \Theta(\sqrt{k})$ due to our
    choice of $q$ and using Corollary~\ref{coro:normal-spectrum},
    and $s_{k+1}(P) = 0$. So, $s_{k+1}(P') = 0$ because $E_P$ is
    in the same subspace as $P$. Now, using Lemma~\ref{lem:least-singular},
    we know that $s_k(P') \geq s_k(P) - \|E_P\| \geq \Omega(\sqrt{k}) > 0$.
    This means that
    $P'$ has rank $k$, so the subspaces spanned by $\Sigma_k$
    and $P'$ are the same.

    As before, we will try to
    bound the distance between the subspaces spanned
    by $P'$ and $\wh{P}$. Note that using Lemma~\ref{lem:weyl-singular},
    we know that $s_k(P') \leq s_k(P) + \|E_P\| \leq O(\sqrt{k})$.

    We wish to invoke Lemma~\ref{lem:sin-theta} again. Let $UDV^T$
    be the SVD of $P'$, and let $\hat{U}\hat{D}\hat{V}^T$ be
    the SVD of $\wh{P}$. Now, for a matrix $M$, let $\Pi_M$ denote
    the projection matrix of the subspace spanned by the columns
    of $M$. Define quantities $a,b,z_{12},z_{21}$ as follows.
    \begin{align*}
        a &= s_{\min}(U^T\wh{P}V)\\
            &= s_{\min}(U^TP'V + U^TE_{\wb{P}}V)\\
            &= s_{\min}(U^TP'V) \tag{Columns of $U$ are orthogonal to columns of $E_{\wb{P}}$}\\
            &= s_k(P')\\
            &\in \Theta(\sqrt{k})\\
        b &= \|U_{\bot}^T\wh{P}V_{\bot}\|\\
            &= \|U_{\bot}^TP'V_{\bot} + U_{\bot}^TE_{\wb{P}}V_{\bot}\|\\
            &= \|U_{\bot}^TE_{\wb{P}}V_{\bot}\|
                \tag{Columns of $U_{\bot}$ are orthogonal to columns of $P'$}\\
            &\leq \|E_{\wb{P}}\|\\
            &\leq O(\ell\sqrt{dq})\\
        z_{12} &= \|\Pi_U E_{\wb{P}} \Pi_{V_{\bot}}\|\\
            &= 0\\
        z_{21} &= \|\Pi_{U_{\bot}}E_{\wb{P}}\Pi_V\|\\
            &\leq \|E_{\wb{P}}\|\\
            &\leq O(\ell{\sqrt{dq}})
    \end{align*}

    Using Lemma~\ref{lem:sin-theta}, we get the following.
    \begin{align*}
        \|\text{Sin}(\Theta)(U,\hat{U})\| &\leq \frac{az_{21} + bz_{12}}
                {a^2-b^2-\min\{z_{12}^2,z_{21}^2\}}\\
            &\leq O\left(\ell\sqrt{dk}\right)\\
            &\leq \alpha
    \end{align*}

    This completes our proof.
\end{proof}

\subsection{Boosting}

In this subsection, we discuss boosting of error
guarantees of Algorithm~\ref{alg:approximate}.
The approach we use is very similar to the well-known
Median-of-Means method: we run the algorithm multiple
times, and choose an output that is close to all
other ``good'' outputs. We formalise this in
Algorithm~\ref{alg:approximate-boosted}.

\begin{algorithm}[h!]
\caption{\label{alg:approximate-boosted}DP Approximate Subspace Estimator Boosted
    $\DPASEB_{\eps, \delta, \alpha, \beta, \gamma, k}(X)$}
\KwIn{Samples $X_1,\dots,X_{n} \in \R^d$.
    Parameters $\eps, \delta, \alpha, \beta, \gamma, k > 0$.}
\KwOut{Projection matrix $\wh{\Pi} \in \R^{d \times d}$ of rank $k$.}
\vspace{5pt}

Set parameters:
    $t \gets C_3 \log(1/\beta)$ \qquad $m \gets \lfloor n/t \rfloor$
\vspace{5pt}

Split $X$ into $t$ datasets of size $m$: $X^1,\dots,X^t$.
\vspace{5pt}

\tcp{Run $\DPASE$ $t$ times to get multiple projection matrices.}
\For{$i \gets 1,\dots,t$}{
    $\wh{\Pi}_i \gets \DPASE_{\eps,\delta,\alpha,\gamma,k(X^i)}$
}
\vspace{5pt}

\tcp{Select a good subspace.}
\For{$i \gets 1,\dots,t$}{
    $c_i \gets 0$\\
    \For{$j \in [t]\setminus\{i\}$}{
        \If{$\|\wh{\Pi}_i-\wh{\Pi}_j\| \leq 2\alpha$}{
            $c_i \gets c_i + 1$
        }
    }
    \If{$c_i \geq 0.6t-1$}{
        \Return $\wh{\Pi}_i$.
    }
}
\vspace{5pt}

\tcp{If there were not enough good subspaces, return $\bot$.}
\Return $\bot.$
\vspace{5pt}
\end{algorithm}

Now, we present the main result of this subsection.

\begin{theorem}\label{thm:approximate-boosted}
    Let $\Sigma \in \R^{d \times d}$ be an arbitrary, symmetric, PSD
    matrix of rank $\geq k \in \{1,\dots,d\}$, and let $0 < \gamma < 1$.
    Suppose $\Pi$ is the projection matrix corresponding
    to the subspace spanned by the vectors of $\Sigma_k$.
    Then given
    $$\gamma^2 \in
        O\left(\frac{\eps\alpha^2n}{d^{2}k\ln(1/\delta)}\cdot
        \min\left\{\frac{1}{k},
        \frac{1}{\ln(k\ln(1/\delta)/\eps)}
        \right\}\right),$$
    such that $\lambda_{k+1}(\Sigma) \leq \gamma^2\lambda_k(\Sigma)$,
    for every $\eps,\delta>0$, and $0 < \alpha,\beta < 1$,
    there exists and $(\eps,\delta)$-DP algorithm that takes
    $$n \geq O\left(\frac{k\log(1/\delta)\log(1/\beta)}{\eps} +
        \frac{\log(1/\delta)\log(\log(1/\delta)/\eps)\log(1/\beta)}{\eps}\right)$$
    samples from $\cN(\vec{0},\Sigma)$, and outputs a projection matrix $\wh{\Pi}$,
    such that $\|\Pi-\wh{\Pi}\| \leq \alpha$ with probability
    at least $1-\beta$.
\end{theorem}
\begin{proof}
    Privacy holds trivially by Theorem~\ref{thm:approximate}.

    We know by Theorem~\ref{thm:approximate} that
    for each $i$, with probability at least $0.7$,
    $\|\wh{\Pi}_i-\Pi\| \leq \alpha$. This means
    that by Lemma~\ref{lem:chernoff-add}, with probability
    at least $1-\beta$, at least $0.6t$ of all
    the computed projection matrices are accurate.

    This means that there has to be at least one projection
    matrix that is close to $0.6t-1>0.5t$ of these
    accurate projection matrices. So, the algorithm
    cannot return $\bot$.

    Now, we want to argue that the returned projection
    matrix is accurate, too. Any projection matrix
    that is close to at least $0.6t-1$ projection
    matrices must be close to at least one accurate
    projection matrix (by pigeonhole principle). Therefore,
    by triangle inequality,
    it will be close to the true subspace. Therefore,
    the returned projection matrix is also accurate.
\end{proof}

%\subsection{Putting It All Together}

%Now, we are ready to finish the main theorem of the section.

%\begin{proof}[Proof of Theorem~\ref{thm:approximate}]
%    By Corollary~\ref{coro:privacy} and Lemma~\ref{lem:final-projection},
%    Algorithm~\ref{alg:approximate} is differentially private
%    and accurate.
%\end{proof}

%\break
\addcontentsline{toc}{section}{References}
%\bibliographystyle{alpha}
%\bibliography{biblio}

\printbibliography

%\appendix

%\input{appendices}

\end{document}